\documentclass{article}

\usepackage[nonatbib, final]{neurips_2019}

\usepackage[utf8]{inputenc} 
\usepackage[T1]{fontenc}    
\usepackage{hyperref}       
\usepackage{url}            
\usepackage{booktabs}       
\usepackage{amsfonts}       
\usepackage{nicefrac}       
\usepackage{microtype}      
\usepackage{amsfonts}
\usepackage{amsmath, amsthm, amssymb}
\usepackage{algorithm}
\usepackage{algorithmic}
\usepackage{caption}
\usepackage{graphicx}
\usepackage{color}
\usepackage{hyperref}
\usepackage{parskip}
\usepackage{bm}
\usepackage{float}
\usepackage[table,x11names]{xcolor}
\usepackage{multirow}
\usepackage{comment}
\usepackage{mathtools}
\usepackage{tikz}
\usetikzlibrary{shapes.geometric}

\DeclareSymbolFont{bbold}{U}{bbold}{m}{n}
\DeclareSymbolFontAlphabet{\mathbbold}{bbold}
\newtheorem{proposition}{Proposition} 

\newcommand{\Ngon}[2][]{\vcenter{\hbox{\begin{tikzpicture}
\node[regular polygon,regular polygon sides=#2,draw,minimum size=1cm,#1](#2-gon){};
\foreach \X in {1,...,#2}{\fill (#2-gon.corner \X) circle[radius=2pt];}
\end{tikzpicture}}}}

\newcommand{\Ngondouble}[2][]{\vcenter{\hbox{\begin{tikzpicture}
\node[regular polygon,regular polygon sides=#2,draw, double, minimum size=1cm,#1](#2-gon){};
\foreach \X in {1,...,#2}{\fill (#2-gon.corner \X) circle[radius=2pt];}
\end{tikzpicture}}}}

\title{A Probability Density Theory for Spin-Glass Systems}

\author{%
  Gavin S. Hartnett
  \\
  The RAND Corporation\\
  \texttt{hartnett@rand.org} \\
  \And
  Masoud Mohseni \\
  Google Research\\
  \texttt{mohseni@google.com}
}

\begin{document}

\maketitle

\begin{abstract}
    Spin-glass systems are universal models for representing many-body phenomena in statistical physics and computer science. High quality solutions of NP-hard combinatorial optimization problems can be encoded into low energy states of spin-glass systems. In general, evaluating the relevant physical and computational properties of such models is difficult due to critical slowing down near a phase transition. Ideally, one could use recent advances in deep learning for characterizing the low-energy properties of these complex systems. Unfortunately, many of the most promising machine learning approaches are only valid for distributions over continuous variables and thus cannot be directly applied to discrete spin-glass models. To this end, we develop a continuous probability density theory for spin-glass systems with arbitrary dimensions, interactions, and local fields. We show how our formulation geometrically encodes key physical and computational properties of the spin-glass in an instance-wise fashion without the need for quenched disorder averaging. We  show that our approach is beyond the mean-field theory and identify a transition from a convex to non-convex energy landscape as the temperature is lowered past a critical temperature. We apply our formalism to a number of spin-glass models including  the Sherrington-Kirkpatrick (SK) model, spins on random Erdős-Rényi graphs, and random restricted Boltzmann machines.
\end{abstract}

\section{Introduction}
Spin-glasses are a general class of models which can be used to study complexity in physics, chemistry, biology, computer science, and social sciences \cite{SteinBook}. They also provide a theoretical and phenomenological framework to analyze hard real-world problems in discrete optimization and probabilistic inference over graphical models \cite{MezardBook}. At the heart of such complex phenomena is the emergent behavior that can occur when disordered systems contain many particles, variables, or agents which exert two-body or higher order interactions. Today, there is a fundamental gap in our knowledge of how such non-trivial correlations emerge at low temperatures. For these systems there is a sudden increase in correlations, occurring simultaneously at various scales (up to the overall system size), as the temperature is reduced below a critical threshold. After half a century of intense study, it is not yet fully understood why, or under what conditions, a distribution of small to large clusters of variables can become rigid or frozen below a critical point in a hierarchical or multi-scale fashion in the absence of any obvious symmetries. Additionally, the relaxation time-scales grow exponentially large as a function of the correlation length-scales, which in the worst-case prevents the system from achieving equilibrium in finite time. This phenomenon is at the heart of the hardness of combinatorial optimization problems, such as random K-SAT, near computational phase transitions \cite{MooreBook}.

Our main motivation is to explore the critical and low temperature properties of spin-glass systems that encode practical computational problems, which are typically at the intermediate scales with respect to number of variables, range of physical interactions, and spatial dimensions. The spin-glass formulation of such problems often involves thousands or even millions of variables, which precludes any hope of successfully applying brute-force or \textit{ab initio} methods. A given instance of these problems typically contains considerable structure, with an underlying graph that could have a power-law distribution over the degree of connectivity with a fat tail for variables with many long-range physical interactions. Such realistic spin-glasses are also often in an intermediate zone with respect to their fractal dimensions and their physical and computational properties, and may be thought of as lying between the two well-studied limiting cases of short-range Edwards-Anderson model \cite{MezardBook} and infinite range Sherrington-Kirkpatrick (SK) model \cite{sherrington1975solvable}. Consequently spin-glass representations of most interesting and relevant problems reside in an uncharted territory that is analytically and computationally intractable. Although the disorders for each instance can be considered fixed or "quenched" for the relevant time-scales, the self-averaging assumption in statistical physics nonetheless becomes inadequate. Mean-field techniques which can be otherwise successfully applied to toy model problems such as random energy models \cite{MezardBook}, or $p$-spin models \cite{Castellani_2005}, become invalid as the fluctuations over the mean values are typically large. Moreover, Renormalization Group (RG) techniques \cite{NishimoriBook} are ineffective as these approaches rely on strong symmetry assumptions, which are difficult to setup for a particular problem class, not well-defined in presence of strong inhomogeneities, and usually involve crude and irreversible coarse-graining of the microscopic degrees of freedom.

Recent advances in deep learning open up the possibility that these non-linear and non-perturbative emergent properties of spin-glass systems could be machine-learned. Unfortunately, many of the most promising machine learning approaches, such as gradient-based iterative optimization, are only valid for distributions over continuous variables and thus they either cannot be directly applied to discrete spin-glass systems; or they can be applied at the cost of simply ignoring the fact that the machine learning algorithm was developed specifically for distributions over continuous variables. Despite this, there has been some progress in using neural-network for discovering new phases of matter or accelerating Monte Carlo sampling \cite{albergo2019flow, huang2017accelerated, liu2017self, shen2018self, li2018neural}. Here, we are interested in eventually applying recent techniques in deep generative models, such as normalizing flows \cite{rezende2015variational}, to discrete spin-glass distributions described by the following family of Hamiltonians
\begin{equation}
    \label{eq:Hspinglass}
    H =  -\sum_i h_i s_i - \sum_{i<j} J_{ij} s_i s_j \,,
\end{equation}
where $s_i \in \{-1 ,1\}$ is an Ising spin, $J_{ij}$ is the coupling matrix, and $h_i$ is the external magnetic field, and $i$ is an index which runs from 1 to $N$. The equilibrium properties of these systems at an inverse temperature $\beta = 1/T$ are governed by the Boltzmann distribution $p(s) = e^{-\beta H}/Z_s$, where $Z_s$ is the usual discrete partition function involving the sum over all $2^N$ possible spin configurations:
\begin{equation}
    Z_s =\sum_{ \{s_i \}} e^{-\beta H} \,.
\end{equation}

Our goal is to formulate an alternate version of the spin-glass problem in terms of a new Hamiltonian \textit{density}, $\mathcal{H}_{\beta}(x)$ over real variables $x$, such that we can obtain a probability density over $x$ as a Boltzmann distribution, $p(x) = e^{-\beta \mathcal{H}_{\beta}(x)}/Z_x$, where $Z_x$ is a new, partition function given by an integration over all possible continuous configurations:
\begin{equation}
    \label{eq:partitionfunc_x}
    Z_x = \int \mathrm{d}^N x \, e^{-\beta \mathcal{H}_{\beta}(x)} \,.
\end{equation}
As we will show, the continuous Boltzmann distribution $p(x)$ can be constructed by using the original discrete distribution $p(s)$ as a prior, and taking the conditional distribution $p(x|s)$ such that the $x$ variables are normally distributed around each of the $2^N$ possible spin configurations.

In order to apply Hamiltonian Monte Carlo to energy-based models of learning, a continuous relaxation method for transforming discrete Hamiltonians of the form Eq.~\ref{eq:Hspinglass} into continuous variables was introduced in Ref. \cite{zhang2012continuous}. In this approach a continuous distribution is defined which is in a sense ``dual'' to the original, discrete distribution. In this work, we use such continuous dual distribution to develop a Hamiltonian density formulation of spin-glasses to explore their behaviors near and far from critical point in sufficiently high and low temperatures. Our continuous formulation of spin-glasses is particularly attractive since it provides a geometric encoding of many thermodynamic properties and allows gradient-based optimization techniques to be applied.\footnote{Recently, a distinct continuous formulation was recently introduced in \cite{caravelli2019continuum}. This formulation appears to be quite different from ours.} One of the key features of our approach is that it does not involve quenched disorder averaging, self-averaging of disordered $J_{ij}$ for very large $N$, or the mean-field assumption of small fluctuations around the mean value of physical observable at the thermodynamic limit, nor any replicas will be introduced. Our approach is thus not a typical mean-field theory and can be applied to spin-glass systems of intermediate dimensions and arbitrary distributions of interactions and local fields. However, we demonstrate that under certain additional assumptions our probability density formulation reduces to the known techniques such as mean-field theory or Thouless, Anderson, and Palmer (TAP) formalism. In particular, we evaluate physical properties of spin-glass models including  SK model, 2D Ising model, spins on random Erdős-Rényi graphs, and random restricted Boltzmann machines. In a separate manuscript, we apply our formulation to build and train deep generative spin-glass models via normalizing flows over non-local latent continuous variables \cite{hartnett2020selfsupervised}.

The layout of this paper is as follows. In Sec.~\ref{sec:generaltheory} we describe the dual continuous distribution of general spin-glasses expressed by Eq.~\ref{eq:Hspinglass}, and introduce these distributions in terms of Hamiltonian densities. This represents an energy landscape over (continuous) spin-glass configurations, which we explore for general couplings. We then analyze the low-temperature limit in Sec.~\ref{sec:lowT} and show that the description of the landscape simplifies. In Sec.~\ref{sec:SKexample} we consider the SK model \cite{sherrington1975solvable}, random restricted Boltzmann machines, and and spin-glasses on random Erdős-Rényi graphs as three illustrative examples. We conclude with a discussion in Sec.~\ref{sec:conclusion}. In the appendices, we present in-depth technical proofs and derivations of the main results and also provide an additional example on 2D Ferromagnetic Ising model.

\section{The Hamiltonian density \label{sec:generaltheory}}
In this section, we employ continuous relaxation method introduced in \cite{zhang2012continuous} to show that the thermodynamic properties of discrete spin-glass systems of the form Eq.~\ref{eq:Hspinglass} may be equivalently formulated in terms of a probability density over a continuous random variable $x \in \mathbb{R}^N$. This distribution is itself a Boltzmann distribution, $p(x) = e^{-\beta \mathcal{H}_{\beta}(x)}/Z_x$, with continuous partition given by Eq.~\ref{eq:partitionfunc_x} and where $\mathcal{H}_{\beta}(x)$ is the \textit{Hamiltonian density} given by\footnote{We call $\mathcal{H}_{\beta}(x)$ a density in order to emphasize how it transforms under a change of variable. For a change of variable $x'=x'(x)$, the Hamiltonian density transforms as $\mathcal{H}_{\beta}(x') = \mathcal{H}_{\beta}(x) - \ln \det\left(\partial x'/\partial x \right)/\beta$, where $\partial x'/\partial x$ is the Jacobian matrix. The term density is \textit{not} meant to indicate that $\mathcal{H}_{\beta}(x)$ is a per-site quantity.}
\begin{equation}
    \label{eq:H_of_x}
    \mathcal{H}_{\beta}(x) := \frac{1}{2} \sum_{i,j} \tilde{J}_{ij} x_i x_j - \frac{1}{\beta} \sum_{i=1}^N \ln [2\cosh\left(\beta  \tilde{h}_i(x) \right)] \,.
\end{equation}
Here we have introduced a shifted coupling matrix,
\begin{equation}
    \tilde{J}_{ij} := J_{ij} + \Delta \, \delta_{ij} \,,
\end{equation}
where $\Delta$ is a shift parameter chosen to ensure that the probability density is integrable, and $\tilde{h}_i(x)$ is an effective local field given by
\begin{equation}
    \label{eq:heff}
    \tilde{h}_i(x) := \sum_j \tilde{J}_{ij} x_j + h_i \,.
\end{equation}

At large radius $||x||_2 \rightarrow \infty$, the first term in $\mathcal{H}_{\beta}(x)$ dominates and the energy landscape is quadratic in $x$. The probability density will therefore be integrable if $\Delta$ is chosen to ensure that the shifted coupling matrix is positive-definite, which is achieved for any $\Delta > \Delta_{\text{min}}$, with
\begin{equation}
    \label{eq:Delta}
    \Delta_{\text{min}} := \max(0, - \lambda_{1}(J) ) \,,
\end{equation}
where $\lambda_{1}(J) \le ... \le \lambda_{N}(J)$ denote the ordered eigenvalues of the matrix $J$. In other words, if $J_{ij}$ is already positive definite, then $\Delta$ can be set to zero. Otherwise, $\Delta$ must be at least as large as the smallest eigenvalue of $J$. Throughout this work, we will set $\Delta = \max(0, \epsilon - \lambda_{1}(J))$, where $0 < \epsilon \ll 1$ is an arbitrarily small positive number meant to ensure positive definiteness, rather than positive semi-definiteness.

The central starting assumption used in the derivation of Eq.~\ref{eq:partitionfunc_x} is that the distribution of the continuous variable $x$, given a discrete Ising spin configuration $s$, be given by a multi-variate Gaussian centered around $s$ and with covariance matrix proportional to the inverse of the shifted coupling matrix, i.e.
\begin{equation}
    p(x|s) = \mathcal{N}(s, \Sigma) \,,
\end{equation}
with mean and covariance matrix given by
\begin{equation}
\label{eq:GaussianParam}
\mu = s \,, \qquad \Sigma = ( \beta \tilde{J} )^{-1} \,,
\end{equation}
Therefore, the joint distribution $p(x, s) = p(x|s) p(s)$ is itself a multi-variate Gaussian with a prior given by the original discrete distribution. Moreover, the above choice of covariance matrix has the important property that the quadratic term in $s$ vanishes in the exponent of the joint distribution expression:
\begin{equation}
    \label{eq:joint}
    p(x,s) = Z_x^{-1} \exp\left[ -\beta \left( \frac{1}{2} x^T \tilde{J} x - s^T \tilde{J} x - h^T s \right) \right] \,,
\end{equation}
where we have combined the normalization of the Gaussian with the discrete partition function in order to form $Z_x$, and where we have employed matrix notation for convenience. This allows $s$ to be trivially marginalized over, which leads to $p(x) = e^{-\beta \mathcal{H}_{\beta}(x)}/Z_x$, with $\mathcal{H}_{\beta}(x)$ given by Eq.~\ref{eq:H_of_x} above. In general, temperature-dependent interactions can arise when degrees of freedom are integrated out, and this happens here as well - hence the subscript.

Similarly, the conditional distribution $p(s|x)$ factorizes over each site, ${p(s|x) = \prod_{i=1}^N p(s_i | x)}$, with
\begin{equation}
    \label{eq:s_given_x}
    p(s_i|x) = \frac{\exp\left(\beta \tilde{h}_i(x) s_i  \right)}{2\cosh\left( \beta \tilde{h}_i(x) \right)} \,.
\end{equation}
Just as the joint distribution may be interpreted as a mixture of Gaussians by writing $p(x, s) = p(x|s) p(s)$, the above expression allows for an additional interpretation where the joint distribution is given by a product of the sigmoidal-like per-site distributions with a prior given by the continuous probability density, i.e. $p(x,s) = \prod_i p(s_i|x) p(x)$.

Rather than starting with a Gaussian conditional distribution and then calculating the joint distribution, an alternative derivation would be to employ the Hubbard-Stratonovich transformation with $x$ now interpreted as the integration variable. The joint distribution may then be defined as the product of the original distribution and the Gaussian integrand:
\begin{align}
    p(s) \int \mathrm{d}^N x \, \frac{ \exp\left(- \frac{1}{2} \left(x-\mu\right)^T \Sigma^{-1} \left(x-\mu\right) \right)}{\sqrt{(2\pi)^N \det \Sigma}} \triangleq \int \mathrm{d}^N x \, p(x,s) \,.
\end{align}
Typically, the Hubbard-Stratonovich transformation is associated with mean-field and/or the replica approach and is applied only after averaging over the disorder. Here, no disorder average has been performed, no replicas have been introduced, and no approximations have been made.

There is a probabilistic map between the continuous and discrete formulations. Importantly, both conditional distributions $p(x|s)$ and $p(s|x)$ are easy to sample from. Given a collection of discrete spin configurations $s$, perhaps obtained through Monte Carlo techniques, a corresponding collection of continuous spin distributions may be obtained using the conditional distribution $p(x|s)$ (and vice versa). The free energies of each formulation may be related to one another via
\begin{equation}
    \label{eq:Zx}
    \ln Z_x  = \ln Z_s + \frac{N}{2} \ln (2\pi) - \frac{1}{2} \ln \det (\beta \tilde{J}) + \frac{N \beta \Delta}{2} \,.
\end{equation}
This expression may be derived by equating the joint distribution $p(x,s)$ (Eq.~\ref{eq:joint}) with $p(x|s) p(s)$. The second and third terms are simply due to the normalization of the multi-variate Gaussian in $p(x|s)$, and last term is due to the fact that $s^T s = N$ for Ising spins. The fact that the continuous and discrete formulations are related in this way indicates that there is no ``free lunch'' here - one cannot use the continuous formulation to circumvent the problems associated with complex spin-glass distributions - for example, the hardness in sampling or the evaluation of the partition function.

The partition function is a generating function for the $n$-point correlation functions. Denoting the usual thermodynamic ensemble over discrete spins as $\langle \cdot \rangle_s := \sum_{\{ s_i \}} \left( e^{-\beta H} \cdot \right)/Z_s$, and the analogous ensemble over continuous configurations as $\langle \cdot \rangle_x := \int \mathrm{d}^N x \left( e^{-\beta \mathcal{H}_{\beta}(x)} \cdot \right)/Z_x$, then by applying $\partial_{h_{i_1}} ... \partial_{h_{i_p}}$ to each side of Eq.~\ref{eq:Zx}, the connected correlation functions of the two ensembles may be related through:
\begin{equation}
    \label{eq:correlation}
    \langle s_{i_1} ... s_{i_p} \rangle_{s, C} = \Big\langle \tanh ( \beta \tilde{h}_{i_1}(x) ) ... \tanh( \beta \tilde{h}_{i_p}(x)) \Big\rangle_{x, C} \,,
\end{equation}
where all indices are assumed distinct and the $C$ subscript denotes connected. \footnote{Since this relation holds for all $p$, it follows that Eq.~\ref{eq:correlation} also holds for the unconnected correlation functions (i.e. the subscript $C$ may be dropped):
$$ \langle s_{i_n} ... s_{i_p} \rangle_{s} = \Big\langle \tanh ( \beta \tilde{h}_{i_1}(x) ) ... \tanh( \beta \tilde{h}_{i_p}(x)) \Big\rangle_{x} \,. $$} In particular, the average local magnetization at site $i$ is related to the continuous variable via $\langle s_{i} \rangle_{s} = \langle \tanh(\beta \tilde{h}_i(x)) \rangle_{x}$. The marginal probability of the spin pointing up at site $i$ is
\begin{equation}
    \label{eq:marginal}
    p(s_i = \pm 1) = \frac{1}{2} \left(1 \pm \Big\langle \tanh(\beta \tilde{h}_i(x)) \Big\rangle_{x} \right) \,.
\end{equation}
This expression allows for an interpretation of the effective field $\tilde{h}_i(x)$ as a global input signal that determines the local spin polarization after averaging over all possible $x$ configurations. In this expression, the hyperbolic tangent plays the role of an activation function, commonly used in artificial neural networks, that determines the polarization of the spin.

We can express the overlap distribution of the original discrete system in terms of the continuous variable using Eq.~\ref{eq:correlation}. If the thermodynamic (Gibbs) measure decomposes into a sum over pure states, each with weight $w_{\alpha}$, then the disorder-dependent overlap distribution is
\begin{equation}
    P_J(q) = \sum_{\alpha \beta} w_{\alpha} w_{\beta} \delta(q_{\alpha \beta}- q) \,.
\end{equation}
This distribution can be regarded as the order parameter of mean field spin-glasses in our continuous formulation, and the moments of this distribution may be expressed in terms of spin correlation functions as \cite{dotsenko2005introduction}:
\begin{equation}
    q^{(p)}_J := \int \mathrm{d} q \, P_J(q) q^p = \frac{1}{N^p} \sum_{i_1 ... i_p} \langle s_{i_1} ... s_{i_p} \rangle_s^2 \,.
\end{equation}
By using Eq.~\ref{eq:correlation}, this may be equivalently written as:
\begin{equation}
    q_J^{(p)} = \frac{1}{N^p} \sum_{i_1 ... i_p} \Big\langle \tanh( \beta \tilde{h}_{i_1}(x) ) ... \tanh( \beta \tilde{h}_{i_p}(x)) \Big\rangle_{x}^2 \,.
\end{equation}
This relation shows how the spin-glass order parameter is encoded in the continuous formulation.

This concludes the derivation of our continuous formulation. One important aspect of using continuous variables is that they provide a geometric encoding of the problem. In particular, $p(s_i |x)$ encodes the likelihood that a given spin will point up or down for a given point in $\mathbb{R}^N$. This probability is in turn determined by the inverse temperature and the strength of the effective local field at that point, $\tilde{h}_i(x)$. The contours of constant $\tilde{h}_i(x)$ are given by shifted ellipsoids, with the shift given by the external local field $h_i$ and the shape and scale of the ellipsoid determined by the $\beta \tilde{J}$. The conditional distribution $p(s_i | x)$ can be used to obtain the marginal probability distribution $p(s_i)$ by integrating over all $\mathbb{R}^N$ and weighting each point according to it probability under the continuous Boltzmann distribution $p(x)$. The $S$-shaped activation function that appears in Eq.~\ref{eq:marginal} implies that the spins will be frozen if the regions of large local effective field are assigned a low energy in the energy landscape given by $\mathcal{H}_{\beta}(x)$. In subsequent sections, we will explore the geometric structure of this landscape further, both for general coupling matrices $\tilde{J}$, and for some well-known examples such as the SK model.

\section{Geometry of the energy landscape \label{sec:landscape}}
The probability density formulation affords several advantages over the original discrete formulation as well as some additional mathematical subtleties. Continuous variables allow alternative sampling methods such as Hamiltonian Monte Carlo \cite{duane1987hybrid} to be applicable, and indeed, this was one of the motivations for the continuous relaxation method \cite{zhang2012continuous}. Another benefit is that the continuous formulation provides a geometric encoding of of the combinatorial optimization problems which may be represented in terms of spin-glass systems. In this section we will derive basic properties of the geometry of the energy landscape $\mathcal{H}_{\beta}(x)$ for arbitrary values of the couplings and graph topology. Later, in Sections \ref{sec:SKexample}, \ref{sec:RBM}, and \ref{sec:randomgraph} we will further explore our formulation of the the SK model, random restricted Boltzmann machines, and spin on random Erdős-Rényi graphs as specific examples.

One of our main results is that the energy landscape defined by $\mathcal{H}_{\beta}(x)$ is convex above a disorder-dependent critical temperature $T_{\text{convex}}$, given in terms of the largest eigenvalue of the shifted coupling matrix:
\begin{equation}
    \label{eq:convex}
    T_{\text{convex}} := \lambda_N(J) + \Delta \,.
\end{equation}
The proof is given in Appendix \ref{sec:convex}. One of the most important mathematical properties of the energy landscape is whether it is convex or not. In particular, convexity of $\mathcal{H}_{\beta}(x)$ implies that the log probability density $p(x)$ is log-concave, and log-concave probability densities enjoy a number of useful properties, such as the fact that the cumulative distribution function (CDF) is also log-concave, as well as the fact that the marginal density over any subset of the $x_i$ variables will also be concave. Convexity of $\mathcal{H}_{\beta}(x)$ also implies practical consequences, for example it means that certain algorithms such as adaptive rejection sampling may be used to \emph{efficiently} sample $p(x)$ \cite{gilks1992adaptive}.

As the temperature is lowered past $T_{\text{convex}}$, the Hamiltonian density becomes non-convex. In order to understand this transition, it will be useful to set the external magnetic field to zero, $h=0$. We first note that the expression for $\mathcal{H}_{\beta}(x)$ in Eq.~\ref{eq:H_of_x} is the sum of two terms. The first is quadratic in $x$, and is guaranteed to be positive for any $x$ since $\Delta$ was chosen to make $\tilde{J}$ positive-definite. Conversely, the second term is negative for any $x$, and it scales linearly in $x$ at large radii, i.e. as $||x||_2 \rightarrow \infty$. Thus, at large radii the first term dominates and the Hamiltonian density is:
\begin{equation}
    \mathcal{H}_{\beta}(x) \sim \frac{1}{2} x^T \tilde{J} x \,,
\end{equation}
which ensures that $p(x)$ is integrable. In contrast, near the origin $x=0$ the expression simplifies to
\begin{equation}
    \mathcal{H}_{\beta}(x) \sim \text{const} + \frac{1}{2} x^T (\tilde{J} - \beta \tilde{J}^2 ) x 
    \,.
\end{equation}
The linear term in the expansion vanishes, and therefore the origin is a critical point of the Hamiltonian density. If $(\tilde{J} - \beta \tilde{J}^2)$ is also positive-definite, then $x=0$ is a minimum. This condition is equivalent to $T > T_{\text{convex}}$, and so in this case $x=0$ is the unique global minimum. As $T$ is lowered below $T_{\text{convex}}$,  the matrix $(\tilde{J} - \beta \tilde{J}^2)$ develops negative eigenvalues, and $x=0$ becomes a saddle.

In addition to the origin becoming unstable, the convex/non-convex transition is also characterized by the appearance of a pair of additional critical points. The critical points of $\mathcal{H}_{\beta}(x)$ solve
\begin{equation}
    \label{eq:criticalpointeq}
    x = \tanh(\beta \tilde{J} x) \,,
\end{equation}
As the temperature approaches $T_{\text{convex}}$ from below, any critical points that exist will merge with the critical point $x=0$, since $x=0$ is the sole critical point for $T > T_{\text{convex}}$. We may therefore linearize the critical point equation around $x=0$. In this case, Eq.~\ref{eq:criticalpointeq} simplifies to $x = \beta \tilde{J} x$. A non-trivial solution of this equation is just an eigenvector of $\beta \tilde{J}$ with eigenvalue 1, which corresponds to $T = T_{\text{convex}}$.  If $v_i^{(N)}$ is the largest eigenvector of $\beta \tilde{J}$ with corresponding eigenvalue $\lambda_{(N)}$, then so is $c \, v_i^{(N)}$ for any non-zero $c$ - in other words the scale is not fixed in the linear treatment. Going beyond linear order will fix $c$ up to a $\mathbb{Z}_2$ reversal $c \rightarrow - c$, since $h=0$. Thus, a pair of critical points will appear as $T_{\text{convex}}$ is reached from above.

The convex/non-convex transition experienced by the continuous distribution $p(x)$ has no counter-part in the original discrete distribution $p(s)$. For every example we study below, $T_{\text{convex}}$ does \textit{not} correspond to a phase transition in the discrete system. In fact, $T_{\text{convex}}$ may be varied without changing the physical content of the theory by using a shift larger than the minimum, i.e. $\Delta > \Delta_{\text{min}}$. However, there is some physical significance of the minimal value of $T_{\text{convex}}$, which can be seen by noting that $\lambda_N(J)$ is the critical temperature predicted by the naive mean-field equation
\begin{equation}
    \label{eq:meanfield}
    x = \tanh\left(\beta J x\right) \,.
\end{equation}
Defining $T_{\text{mean-field}} := \lambda_N(J)$, we may then write
\begin{equation}
    \label{eq:temperature_relation}
    T_{\text{convex}} = T_{\text{mean-field}} + \Delta \,.
\end{equation}
Moreover, if the eigenvalues of $J$ lie in a symmetric interval, with $\lambda_N(J) = - \lambda_1(J)$, then $\Delta_{\text{min}} = \lambda_N(J)$, and $T_{\text{convex}} \ge T_{\text{mean-field}} + \Delta_{\text{min}} = 2 \, T_{\text{mean-field}}$. With our choice of $\Delta = \max(0, \epsilon - \lambda_1(J))$ we have that $T_{\text{convex}} = 2 \, T_{\text{mean-field}} + \epsilon$ so that as $\epsilon \rightarrow 0$ the inequality is saturated.

To summarize the results of this section: as the temperature is lowered past a $T_{\text{convex}}$, the Hamiltonian density becomes non-convex, the critical point at $x=0$ becomes unstable, and a pair of non-trivial critical points with $x \neq 0$ appears. It is difficult to go much further than this description and make more detailed statements about the geometry of the energy landscape without specifying the couplings $J$. This is to be expected, since our formalism applies to all spin-systems of the form Eq.~\ref{eq:Hspinglass}, which includes both spin-glasses and ferromagnetic systems like the 2d Ising model. Below in Sections \ref{sec:SKexample}, \ref{sec:RBM} and \ref{sec:randomgraph} we will further analyze the landscape for the Sherrington-Kirkpatrick model, random Restricted Boltzmann Machines, and spin-glasses on random Erdős-Rényi graphs respectively. The case of 2D ferromagnetic Ising model system is explored in Appendix \ref{sec:examples}.

\section{Geometry of the energy landscape deep in the spin-glass phase \label{sec:lowT}}
In this section we will discuss the low-temperature limit of our formalism. Our goal will be to provide some insight into the geometry of energy landscapes of systems which are deep in the spin-glass phase (when such a phase exists), and to show how the metastable spin-glass states are geometrically encoded in the Hamiltonian density, $\mathcal{H}_{\beta}(x)$. We will leave the coupling matrices and disorder distribution unspecified, and as a result, our discussion will be somewhat general.

We begin by taking the low-temperature expansion of the Hamiltonian density: ${\mathcal{H}_{\beta}(x) = \mathcal{H}_{\infty}(x) + \mathcal{O}(\beta^{-1})}$, where
\begin{equation}
    \mathcal{H}_{\infty}(x) := \frac{1}{2} x^T \tilde{J} x - \sum_{i=1}^N |\tilde{J} x|_i \,.
\end{equation}
The equation governing the zero critical points may be obtained from $\mathcal{H}_{\infty}(x)$ directly or from the $\beta \rightarrow \infty$ limit of Eq.~\ref{eq:criticalpointeq}:
\begin{equation}
    x = \text{sgn}( \tilde{J} x )\,.
\end{equation}
Additionally, the Hessian (matrix of second derivatives) of $\mathcal{H}_{\infty}(x)$ is simply the shifted coupling matrix: $K_{\infty} := \tilde{J}$. There is a subtlety here, which is that the Hessian is not defined for points which satisfy $(\tilde{J} x)_i = 0$ for any $i$ because the absolute value function is not differentiable at the origin. This is an important observation, since without it one would conclude that $\mathcal{H}_{\infty}(x)$ is convex, which it certainly is not.

With these ingredients, the integral defining the partition function $Z_x$ may then be formally written as a sum over the critical points using Laplace's method:
\begin{align}
    \label{eq:saddlept2}
    Z_x = \int \mathrm{d}^N x \, e^{-\beta \mathcal{H}_{\beta}(x)}
    &\approx \sum_{\alpha} e^{-\beta \mathcal{H}_{\infty}(x^{(\alpha)})} \int \mathrm{d}^N x \, e^{- \frac{\beta}{2} (x - x^{(\alpha)})^T \tilde{J} (x - x^{(\alpha)})} \nonumber \\
    &= \sqrt{\frac{(2\pi)^N}{\det(\beta \tilde{J})}} \sum_{\alpha} e^{-\beta \mathcal{H}_{\infty}(x^{(\alpha)})} \,,
\end{align}
where $x^{(\alpha)}$ are the critical points of the Hamiltonian density, and the prefactor is due to the Gaussian integration around each critical point. In writing the above expression we have assumed that all critical points are minima, and so the Gaussian integration converges. Without specifying the coupling matrix it is difficult to say much about the existence or non-existence of saddles, beyond the fact that $x=0$ is always both a solution of the critical point equation and a point for which the Hessian is not defined. This and any other similar points will require some special treatment, for example by rotating the integration contours and including sub-leading corrections in $\beta^{-1}$. Ignoring such complications, general correlation functions may also be formally written as a sum over critical points as:
\begin{equation}
    \left\langle f(x) \right\rangle_{x} \approx \sum_{\alpha} \omega_{\alpha} f(x^{(\alpha)}) \,,
\end{equation}
where $\omega_{\alpha} := e^{-\beta \mathcal{H}_{\beta}(x^{(\alpha)})}/Z_x$ is the Boltzmann weight of each critical point, and $f$ is an arbitrary function.\footnote{We have ignored the Gaussian prefactor here, since 1) it is subleading in $\beta^{-1}$, and 2) all critical points receive the same prefactor since the Hessian is just the constant matrix $\tilde{J}$.} Thus, the critical points can be seen to encode almost all of the physics of the problem in the low-temperature limit. Applying the saddle-point method to the 1-point function (i.e. the $p=1$ case of Eq.~\ref{eq:correlation}), the saddle point coordinates are related to the average per-site magnetizations via
\begin{equation}
    \label{eq:onepoint_saddle}
    \langle s_i \rangle_{s} \approx \sum_{\alpha} \omega_{\alpha} x_i^{(\alpha)}  \,.
\end{equation}
Therefore, in our continuous formulation the critical points are very analogous to the pure states of spin-glass theory. Pure states of spin-glasses are sub-regions in the state space which are separated by large energy barriers, and the system is sub-ergodic in those regions even though global ergodicity is broken \cite{mezard1987spin}. Indeed, if the sum over critical points is restricted to just a single critical point (or if there is only one such dominant critical point in the thermodynamic limit), then all connected correlation functions vanish, for example
\begin{equation}
    \langle x_{i_1} x_{i_2} \rangle_x = \langle x_{i_1} \rangle_x \langle x_{i_2} \rangle_x \,.
\end{equation}
This property is also known as cluster decomposition.

The pure states have a simple geometric interpretation in our formalism. Recall that the continuous probability density may be written as a weighted sum of Gaussians, each centered around one of the $2^N$ spin configurations, $p(x) = \sum_{\{ s \}} p(x|s) p(s)$. The covariance matrix of each Gaussian is $\Sigma = (\beta \tilde{J})^{-1}$, and so the level sets of $p(x|s)$ are $N$-dimensional ellipsoids whose shape is determined by the eigenvectors and eigenvalues of $\Sigma$. In general, the density clouds due to each spin configuration will overlap - for example above $T_{\text{convex}}$ the Gaussians are so broad that the resulting continuous distribution $p(x)$ is log-concave. At low-temperatures, the distribution will ``fragment'' into a number of distinct modes. An example of this is shown in Fig.~\ref{fig:lowTcontourplot} for the simple case of just two spins, $N=2$.

\begin{figure}
\centering
  \includegraphics[width=1\linewidth]{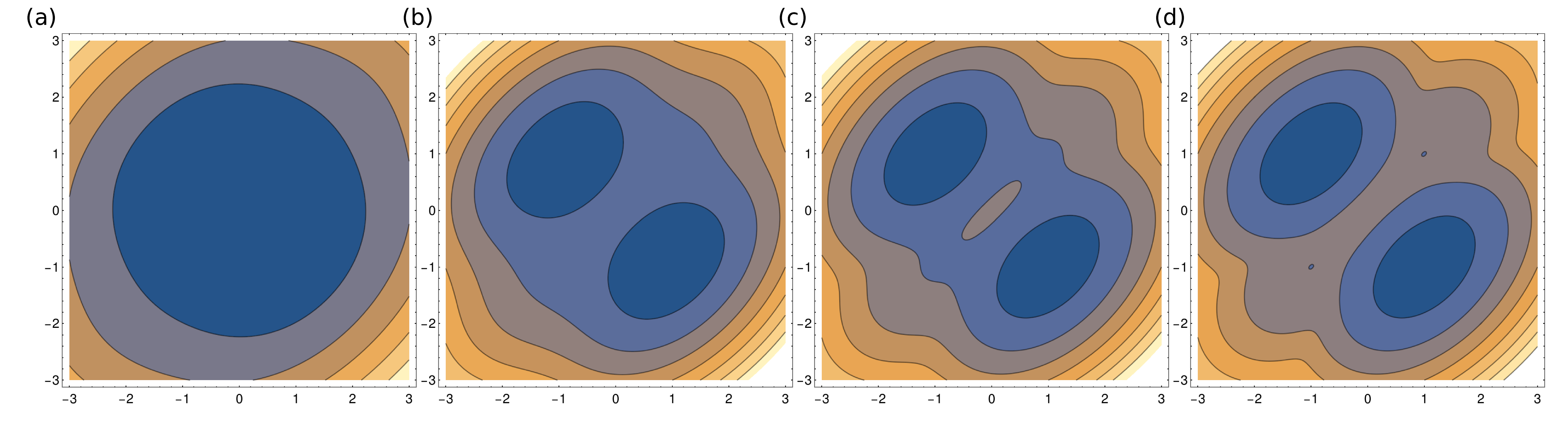}
  \captionof{figure}{
  Contour plot of the Hamiltonian density for a system of two spins with $J_{12} = J_{21} = 0.79257$ with shift $\Delta = \max(0, \epsilon - \lambda_1(J))$ and $\epsilon = 0.1$ for
  (a) $T = T_{\text{convex}}$
  (b) $T = T_{\text{convex}}/2$
  (c) $T =  T_{\text{convex}}/4$
  (d) $T =  T_{\text{convex}}/8$
  . The density clouds due to spin configurations overlap above $T_{\text{convex}}$; that is the Gaussians are sufficiently broad that the resulting continuous distribution $p(x)$ is log-concave. At low-temperatures, the distributions become fragmented into several distinct modes. Blue regions correspond to low energy configurations.
  \label{fig:lowTcontourplot}}
\end{figure}

The nature of this fragmentation depend on how the $\beta \rightarrow \infty$ limit is taken. Suppose that the original coupling matrix $J$ has both positive and negative eigenvalues, so that the eigenvalues of the shifted matrix $\tilde{J}$ satisfy $\lambda_i(\tilde{J}) \ge \epsilon$ (recall that the purpose of introducing the small positive constant $\epsilon$ was to guarantee positive-definiteness of $\tilde{J}$). If $\epsilon$ is chosen to be temperature-independent, then the $\beta \rightarrow \infty$ limit pushes all the eigenvalues of $(\beta \tilde{J})$ to infinity, and consequently all the eigenvalues of ${ \Sigma = (\beta \tilde{J})^{-1} }$ approach zero. In this case $p(x)$ is composed of $2^N$ distinct delta functions with different weights, and the pure states are rather trivially just the $2^N$ spin configurations. However, if instead $\beta \epsilon$ is held fixed as $\beta \rightarrow \infty$, then the eigenvalue spectrum of $\Sigma$ will range from 0 to the finite value $1/(\beta \epsilon)$. Thus, the shape of the ellipsoid defining the level sets of the Gaussians will shrink to a point in some directions, and remain finite in others. In this case the fragmentation of $p(x)$ will be more interesting. Groups of spin configurations will merge to form pure states as determined by the geometry of the zero-temperature ellipsoids in relation to the $2^N$ vertices of the $[-1,1]^N$ hypercube.

The pure states of non-disorder averaged spin-glasses can be associated with solutions of a modified mean-field equation known as the Thouless, Anderson, and Palmer (TAP) equation, which was derived in order to correct the failure of naive mean-field theory to describe the spin-glass phase of the SK model. The naive mean-field equation is given in Eq.~\ref{eq:meanfield}, whereas the TAP equation is given by \cite{thouless1977solution}:
\begin{equation}
    \label{eq:TAP}
    x_i = \tanh\Big( \beta \sum_j J_{ij} x_j - \beta^2 x_i \sum_j J_{ij}^2 (1-x_j^2) \Big) \,.
\end{equation}
We have argued that the critical points of the continuous probability density may be interpreted as pure states at low temperature, and thus there should be a connection between these and the solutions of the TAP equation. Here we will establish such a connection at zero temperature, for which the TAP equation simplifies considerably: $x = \text{sgn}\left( J x \right)$. Importantly, this is also the zero-temperature limit of the naive mean-field equation Eq.~\ref{eq:meanfield}. The zero-temperature limit of the TAP/naive mean-field equations may be compared with the equation governing the critical points of $\mathcal{H}_{\infty}(x)$:
\begin{subequations}
\begin{equation}
    x = \text{sgn}(\tilde{J} x)\,, \qquad \text{critical points of } \mathcal{H}_{\infty}(x) \,.
\end{equation}
\begin{equation}
    x = \text{sgn}\left( J x \right) \,, \qquad \text{naive mean-field/TAP} \,.
\end{equation}
\end{subequations}
Note that the TAP/naive mean-field equation depends on the original coupling matrix $J$, whereas the critical points of the Hamiltonian density depend on the shifted coupling matrix $\tilde{J} = J + \Delta \, \mathbbold{1}_{N \times N}$. A key result is that solutions of the zero-temperature naive mean-field equation/TAP equation are also critical points of the zero temperature Hamiltonian density:
\\
\begin{proposition}
\label{prop:subset}
If $x$ is a solution of the zero temperature TAP equation $x = \text{sgn}(J x)$, then $x$ is also a solution of the zero temperature critical point equation $x = \text{sgn}( \tilde{J} x)$.
\end{proposition}

\begin{proof}
Suppose $x = \text{sgn}(J x)$. The result holds for any $\Delta \ge \Delta_{\text{min}} \ge 0$, so we will consider the cases $\Delta = 0$ and $\Delta > 0$ separately. If $\Delta =0$, then clearly the mean-field and critical point equations are identical. If $\Delta > 0$, then $\text{sgn}(\Delta x) = \text{sgn}(x) = \text{sgn}( \text{sgn}(J x)) = \text{sgn}(J x)$. Thus,
\begin{align}
    \text{sgn}( \tilde{J} x ) &= \text{sgn}( J x + \Delta x) = \text{sgn}( \text{sgn}(J x) \, |J x| + \text{sgn}(\Delta x) \, |\Delta x| ) \\
    &= \text{sgn}\left( \text{sgn}(Jx) \left( |J x| + | \Delta x| \right) \right) = \text{sgn}(J x) \nonumber \\
    & = x\nonumber \,.
\end{align}
\end{proof}
This establishes that for $T=0$ every solution of the TAP equation is also a critical point of $H_{\infty}(x)$. The converse does not hold: there are critical points of the Hamiltonian density which are not solutions of the TAP equation. To understand the significance of these points, recall that the nature of the pure states depends on whether $\epsilon$ is held fixed as $\beta \rightarrow \infty$, or if instead $\beta \epsilon$ is held fixed. In the first case, the pure states of the continuous formulation are somewhat trivial, as any of the $2^N$ spin configurations will be a pure state according to the above discussion. There may additionally be critical points with $x_i = (\tilde{J} x)_i = 0$ for some $i$ which will not correspond to any Ising spin configuration. For example, the point $x=0$ is always a critical point. Since the TAP solutions are a subset of all possible spin configurations, the zero-temperature critical points will include both the TAP solutions as well as all other spin configurations and any saddle-like points such as $x=0$. If the zero-temperature limit is instead taken while holding $\beta \epsilon$ fixed, then the critical points will include just a subset of all $2^N$ spin configurations. That subset will include the TAP states, and possibly other spin configurations and saddle-like points.

\section{Sherrington-Kirkpatrick Model \label{sec:SKexample}}
In order to build intuition for the probability density formulation of general spin-glass systems, in this section we consider as an example the Sherrington-Kirkpatrick (SK) model \cite{sherrington1975solvable}. By specifying the coupling matrix $J$ (or rather, the disorder distribution from which $J$ is drawn), we may further explore the geometry of the energy landscape and the nature of both the spin-glass and convexity transitions in our formulation.

The Sherrington-Kirkpatrick (SK) model is defined by specifying that the couplings $J_{ij}$ be drawn from an iid Gaussian distribution \cite{sherrington1975solvable}:
\begin{equation}
J_{ij} \sim \mathcal{N}\left( 0, \frac{\mathcal{J}^2}{N} \right) \,, \qquad (i < j) \,,
\end{equation}
where the $i > j$ values are fixed by symmetry of $J$ to be the same as the $i < j$ values, and the diagonal entries are zero. The coupling parameter $\mathcal{J}$ controls the variance of the disorder. The eigenvalue distribution of $J$ in the large-$N$ limit is simply the Wigner semi-circle distribution, so that the probability density of the eigenvalues of $J$ is
\begin{equation}
p_{J}(\lambda) = \frac{2}{\pi R^2} \sqrt{R^2 - \lambda^2} \, 1_{[-R,R]}(\lambda) \,,
\end{equation}
where $1_{[-R,R]}(\lambda)$ is the indicator function. The radius of the semi-circle is related to the coupling parameter via $R = 2  \mathcal{J}$. Since the eigenvalues of $J$ are restricted to the strip $[-R,R]$, the eigenvalues of the shifted coupling matrix $\tilde{J}$ will be shifted to lie within the strip $[\epsilon, 2R + \epsilon]$. The eigenvalues of both $J$ and $J_{\Delta}$ are depicted in Fig.~\ref{fig:SK_eigs}
\begin{figure}
\centering
  \includegraphics[width=0.6\linewidth]{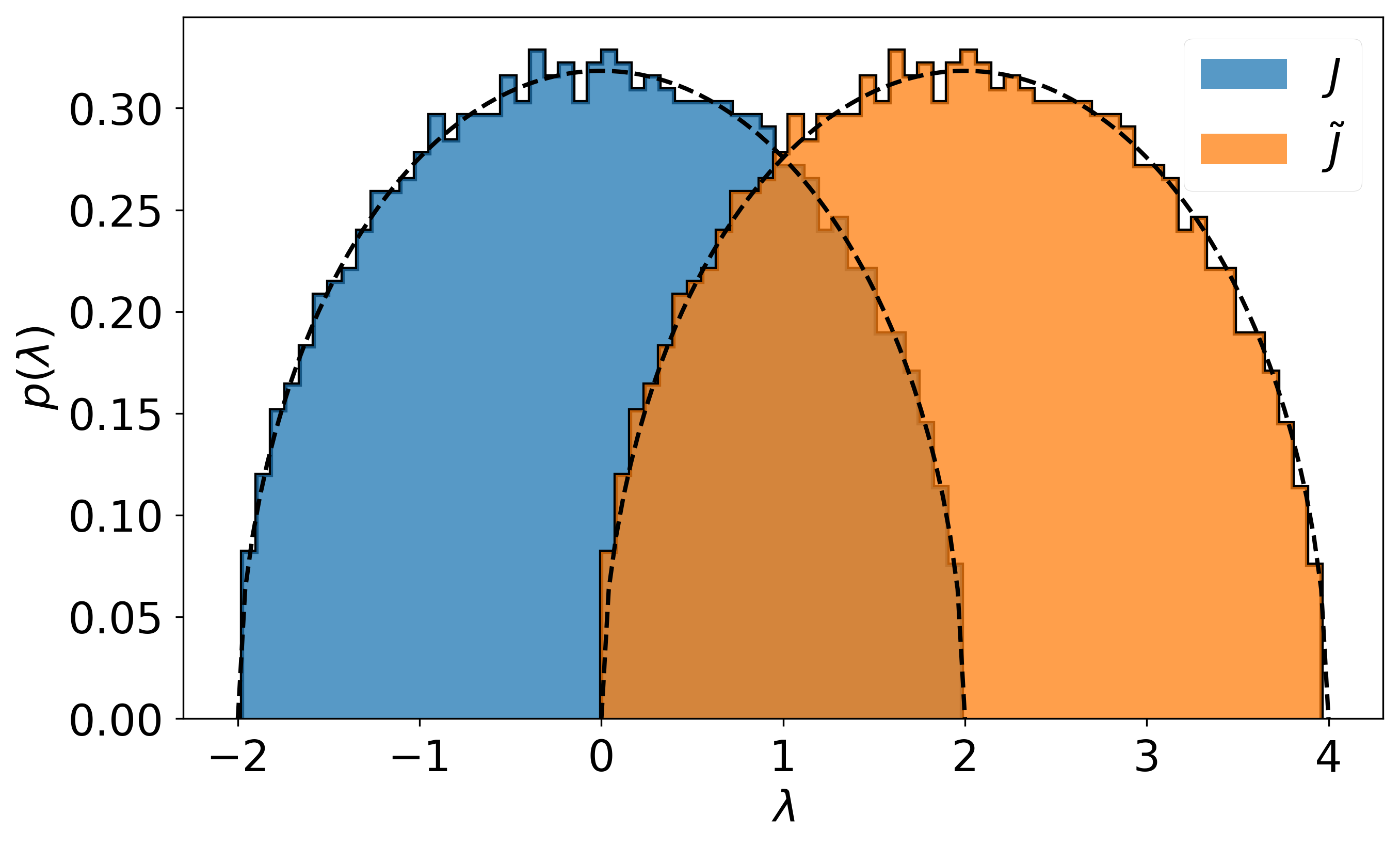}
  \captionof{figure}{
  The eigenvalue distribution of the coupling matrix $J$ and the shifted coupling matrix $\tilde{J}$ for the SK model. Both distributions are described by the Wigner semi-circle distribution, shown in black for both $J$ and $\tilde{J}$. The size of the shift has been chosen so that the shifted distribution has support on the positive real numbers, $\lambda \in (0, \infty)$.
  \label{fig:SK_eigs}}
\end{figure}

Using the radius of the Wigner semi-circle (and disregarding $\epsilon$ for now by setting it to zero), we have
\begin{equation}
T_{\text{mean-field}} = 2 \mathcal{J} \,, \qquad T_{\text{convex}} = 2 \, T_{\text{mean-field}} \,.
\end{equation}
These may be contrasted with the critical temperature below which the system is in a spin-glass phase:
\begin{equation}
    T_{\text{crit}} = \mathcal{J} \,.
\end{equation}
Therefore, we have found that \begin{equation}
    T_{\text{crit}} < T_{\text{mean-field}} < T_{\text{convex}} \,.
\end{equation}
This indicates that the Hamiltonian density becomes non-convex due to the appearance of multiple critical points well before any transition to an ordered phase occurs.\footnote{It is worth noting that these results are strictly only valid for $N \rightarrow \infty$. For finite-$N$ both the eigenvalues of $J$ and the critical temperature will exhibit fluctuations due to finite-sized effects. The fluctuation of the critical temperature due to finite-size effects is investigated in \cite{castellana2011role}.}

\subsection{High-temperature limit}
The fact that $T_{\text{convex}} \neq T_{\text{crit}}$ is intriguing. Naively one might have thought that the two temperatures would have coincided because the transition from a convex to non-convex Hamiltonian density represents a real and significant change in the corresponding Boltzmann distribution. Moreover, the minimal value of $T_{\text{convex}} = 4 \mathcal{J}$ does not appear to have been previously identified as having any particular importance for the well-studied SK model. The mathematical transformation from the original discrete variables to the continuous variables was exact and involved no approximation; however, one still needs to verify that spin-glass transition has not been shifted and still occurs at $T_{\text{crit}}$ not at $T_{\text{convex}}$ when the convexity is no longer guaranteed. To this end, we carried out a high-temperature expansion in terms of the continuous variables and find exact agreement with the expansion in terms of the original discrete variables carried out by Thouless, Anderson, and Palmer in \cite{thouless1977solution}. In both cases, the expansion breaks down at the spin-glass phase transition temperature $T=T_{\text{crit}}$ and not at the higher temperature $T_{\text{convex}}$. We will provide an outline of the calculation here, and a more detailed treatment can be found in Appendix \ref{sec:highTSK}.

Using Eq.~\ref{eq:Zx}, the partition function $Z_s$ may be written in terms of the continuous variables as
\begin{equation}
    Z_s = e^{-\frac{N \beta \Delta}{2}} \left\langle \prod_i \cosh\left( \beta^{1/2} (\tilde{J} x)_i \right) \right\rangle_0 \,.
\end{equation}
The expectation value is taken over a properly normalized Gaussian distribution with zero-mean and covariance matrix $\tilde{J}^{-1}$. A high-temperature expansion may then be performed by expanding around $\beta = 0$. At each order the Gaussian integrals may be performed by Wick contractions, which introduces an increasing number of terms as the order of the expansion increases. The calculation simplifies dramatically if the disorder is averaged over. Denoting the disorder average as $\langle \cdot \rangle_J$, the final result is
\begin{equation}
    \langle \ln Z_s \rangle_J = N\left( \frac{(\beta \mathcal{J})^2}{4} + \ln 2 \right) + \frac{1}{4} \ln\left(1-\beta^2 \mathcal{J}^2 \right) + \text{(non-singular)} + \mathcal{O}(N^{-1}) \,.
\end{equation}
For $T > T_{\text{crit}}$ the sub-extensive terms may be neglected, but as the temperature of the spin-glass transition is approached from above the logarithm becomes singular, indicating a breakdown of the perturbative expansion. Not only is the free energy analytic at the minimal convexity transition temperature $\min_{\Delta} T_{\text{convex}} = 4 \mathcal{J}$, but any dependence on the shift $\Delta$ cancels out, since $\Delta$ was only introduced as part of our formulation.

The above result was first derived in \cite{thouless1977solution} by considering the expansion of $Z_s$ in terms of the original discrete spin variables. In both cases - the expansion in terms of $s$ and the expansion in terms of $x$, the singular logarithm term is obtained by summing an infinite number of terms. In terms of Feynman diagrams, the terms that contribute to the singularity correspond to double-sided regular $n$-gons for $n \ge 3$:
\begin{equation*}
    \Ngondouble{3} + \Ngondouble{4} + \Ngondouble{5} + \Ngondouble{6} + \Ngondouble{7} + \cdots
\end{equation*}
The fact that both expansions agree and yield no non-analyticity at $T_{\text{convex}}$ indicates that the convex/non-convex transition is not associated with a thermodynamic phase transition. It also provides a consistency check that the continuous formulation does not break down below $T_{\text{convex}}$.

\subsection{Zero-temperature limit}
Lastly, we investigated the zero-temperature limit of the SK model by studying the critical points. These are solutions of the equation $x = \text{sgn}(\tilde{J} x)$. We generated a large number of such solutions by randomly initializing $x^{(0)} \in \{-1, 1\}^N$ and then applying the iterative update rule below until a solution was found (or the algorithm failed to converge after a set number of iterations):
\begin{equation}
    x^{(t)} = \frac{1}{2} \left( x^{(t-1)} +  \text{sgn}(\tilde{J} x^{(t-1)}) \right) \,, \quad \text{(critical point)} \,.
\end{equation}
This update rule corresponds to performing gradient descent on $\mathcal{H}_{\beta}(x)$, using a learning rate of $1/2$ and $\tilde{J}^{-1} \nabla \mathcal{H}_{\beta}(x)$ in place of the usual gradient $ \nabla \mathcal{H}_{\beta}(x)$.\footnote{We thank Dan Ish for pointing this out to us.} We also generated a large number of solutions of the zero-temperature mean-field/TAP equation $x = \text{sgn}(J x)$ using the same procedure with update rule given by:
\begin{equation}
    x^{(t)} = \frac{1}{2} \left( x^{(t-1)} + \text{sgn}(J x^{(t-1)}) \right) \,, \quad \text{(mean-field/TAP)} \,.
\end{equation}

In agreement with Proposition \ref{prop:subset} above, we found that every mean-field/TAP solution also solved the critical point equation. Interestingly, we also found that none of the critical point solutions generated this way also solved the mean-field/TAP equation. This is consistent with our earlier observation (that held for large $\Delta)$ that the saddle-like critical points exponentially out-numbered the minima. We also found that the solutions produced by the iterative method applied to each equation had widely separated energies. Fig.~\ref{fig:SK_zeroT_energies} plots the distribution of energies of each set of solutions.\footnote{It should be emphasized that the iterative update rule/gradient descent method we used almost certainly does not generate solutions uniformly. We generated 10,000 unique solutions using each approach, but for the SK model we expect an exponential (in $N$) number of solutions \cite{bray1980metastable} and the solutions plotted in Fig.~\ref{fig:SK_zeroT_energies} may not be representative of the overall distribution. Rather, they are representative of the distribution obtained when solutions are generated using the iterative update rule/gradient descent.}
\begin{figure}
\centering
  \includegraphics[width=0.6\linewidth]{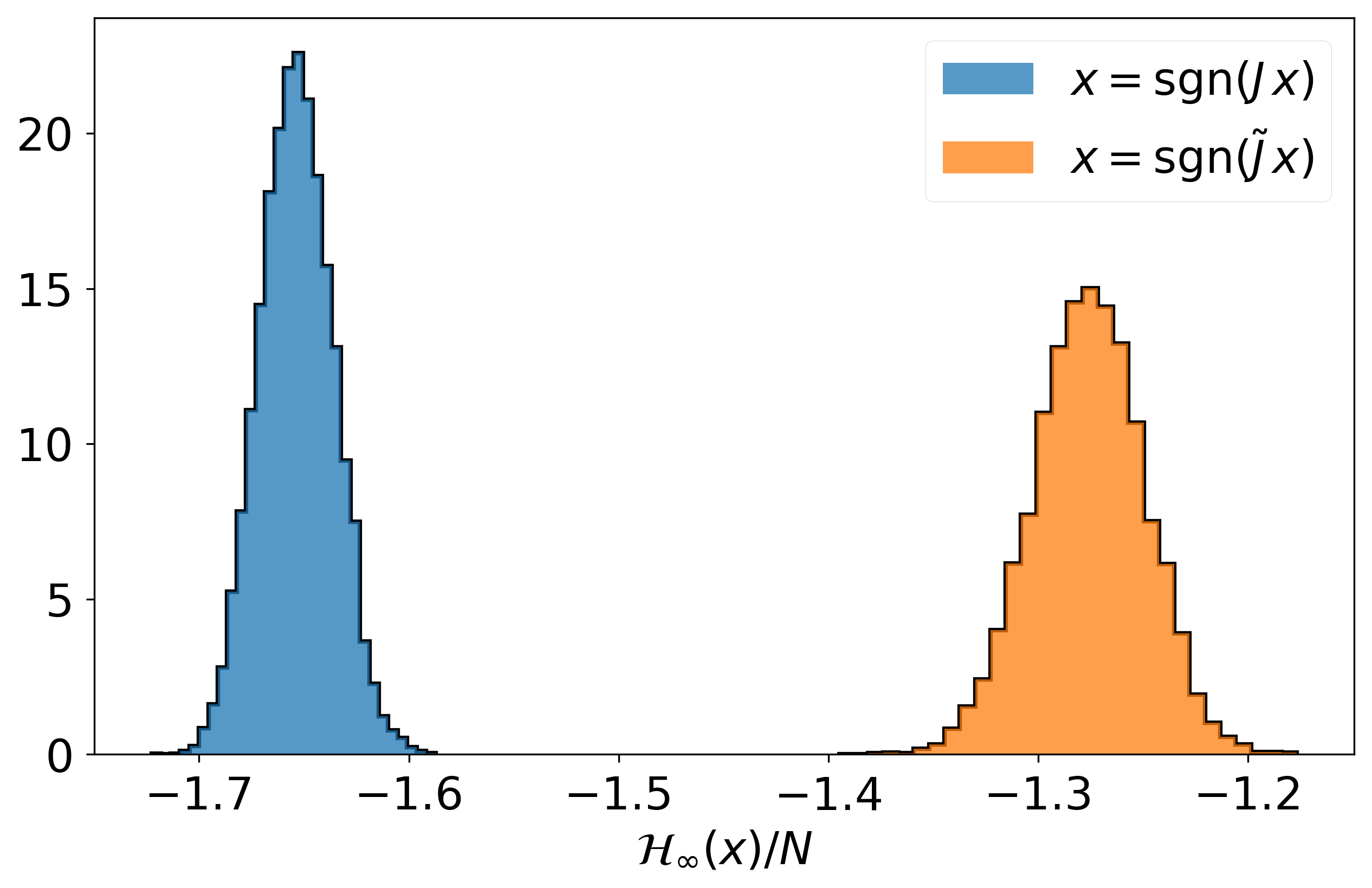}
  \captionof{figure}{
  The distribution of the per-site Hamiltonian densities at zero temperature, i.e. $\mathcal{H}_{\infty}(x)/N$ obtained by an application of the iterative procedure discussed in the text. This procedure was applied to two equations, the mean-field/TAP equation $x = \text{sgn}(J x)$ and the critical point equation $x = \text{sgn}(\tilde{J} x)$, and in both cases $N=500$. Note that every solution of the first equation is also a solution of the second, although the converse is not true. Here we have set $\epsilon = 0$ in $\Delta_{\text{min}}$. This plot shows that the typical critical point has a much higher energy than typical solutions of the mean-field equation (when both sets of solutions are obtained using the iterative procedure).
  \label{fig:SK_zeroT_energies}}
\end{figure}

\section{Random Restricted Boltzmann Machines \label{sec:RBM}}
As a second example, we study the bipartite SK model, which is the natural extension of the SK model to bipartite complete graphs. This example also has significance in machine learning as it represents a randomly initialized Restricted Boltzmann Machine (RBM) \cite{smolensky1986information}. In particular, the bipartite SK model describes a random initialization of RBMs where the biases have been set to zero. The connection between the bipartite SK model and RBMs has been recently studied in \cite{decelle2017spectral, decelle2018thermodynamics, hartnett2018replica}.

In this case, the coupling matrix $J_{ij}$ takes on the block form:
\begin{equation}
J =
\begin{pmatrix}
0 & W \\
W^T & 0
\end{pmatrix} \,,
\end{equation}
with $W$ a $N_v \times N_h$ matrix, where $N_v$ is the number of visible spins and $N_h$ is the number of hidden spins. The total number of spins is $N=N_v+N_h$, and the spin vector may be written as $s^T = (v, h)$. As in the SK model, the weights $W_{ij}$ in the bipartite SK model will  be iid normally distributed:
\begin{equation}
    W_{ij} \sim \mathcal{N}\left(0, \frac{\mathcal{J}^2}{\sqrt{N_v N_h}}\right) \,.
\end{equation}

Using the relation
\begin{equation}
    \det \begin{pmatrix}
A & B \\
C & D
\end{pmatrix} = \det\left(A - B D^{-1} C \right) \det(D)
\end{equation}
for block matrices $A, B, C, D$, the characteristic equation for $J$, $\det\left( J - \lambda \mathbbold{1}_{N \times N} \right) = 0$, is equivalent to the condition
\begin{equation}
    \det \left( W W^T - \lambda^2 \right) = 0\,,
\end{equation}
provided that $\lambda \neq 0$. Thus, the non-zero eigenvalues of $J$ are related to the eigenvalues of $W W^T$ via
\begin{equation}
    \label{eq:JtoWWT}
    \lambda_i \left( J \right) = \pm \lambda_i \left( W W^T \right)^{1/2} \,.
\end{equation}

The eigenvalue distribution of $W W^T$ in the large-$N$ limit defined by $N \rightarrow \infty$ with $\kappa = N_v/N_h$ held fixed is given by the Marchenko–Pastur distribution \cite{marvcenko1967distribution}, which in our conventions is:
\begin{equation}
    p_{W W^T}(\lambda) =  \frac{\sqrt{(R_+ - \lambda)(\lambda - R_-)}}{2\pi \mathcal{J}^2 \kappa^{1/2} \lambda} 1_{[R_-,R_+]}(\lambda) + \max\left(0, 1 - \kappa^{-1} \right) \delta(\lambda) \,,
\end{equation}
where
\begin{equation}
    R_{\pm} = \mathcal{J}^2 \left(\kappa^{-1/4} \pm \kappa^{1/4} \right)^2 \,.
\end{equation}
In Fig.~\ref{fig:Marchenko_Pastur} we plot the eigenvalue distribution for both $W W^T$ and $J$.

As a result of this analysis, we conclude that in the large-$N$ limit $\lambda_i(J) \in [-\sqrt{R_+} , \sqrt{R_+}]$, and also $\lambda_i(\tilde{J}) \in [0, 2\sqrt{R_+}]$ (again neglecting $\epsilon$). Thus, we have that
\begin{equation}
    T_{\text{mean-field}} = \mathcal{J} \left( \kappa^{-1/4} + \kappa^{1/4} \right) \,, \quad \text{and} \quad  T_{\text{convex}} = 2 \, T_{\text{mean-field}} \,.
\end{equation}
Moreover, as in the SK model, both of these temperatures are higher than the critical temperature of the spin-glass phase transition, which in our conventions is \cite{hartnett2018replica}:
\begin{equation}
    T_{\text{crit}} = \mathcal{J} \,.
\end{equation}
Similar to the case of the SK model, the convex/non-convex transition happens well before the phase transition occurs as the temperature is lowered.

\begin{figure}
\centering
  \centering
  \includegraphics[width=0.95\linewidth]{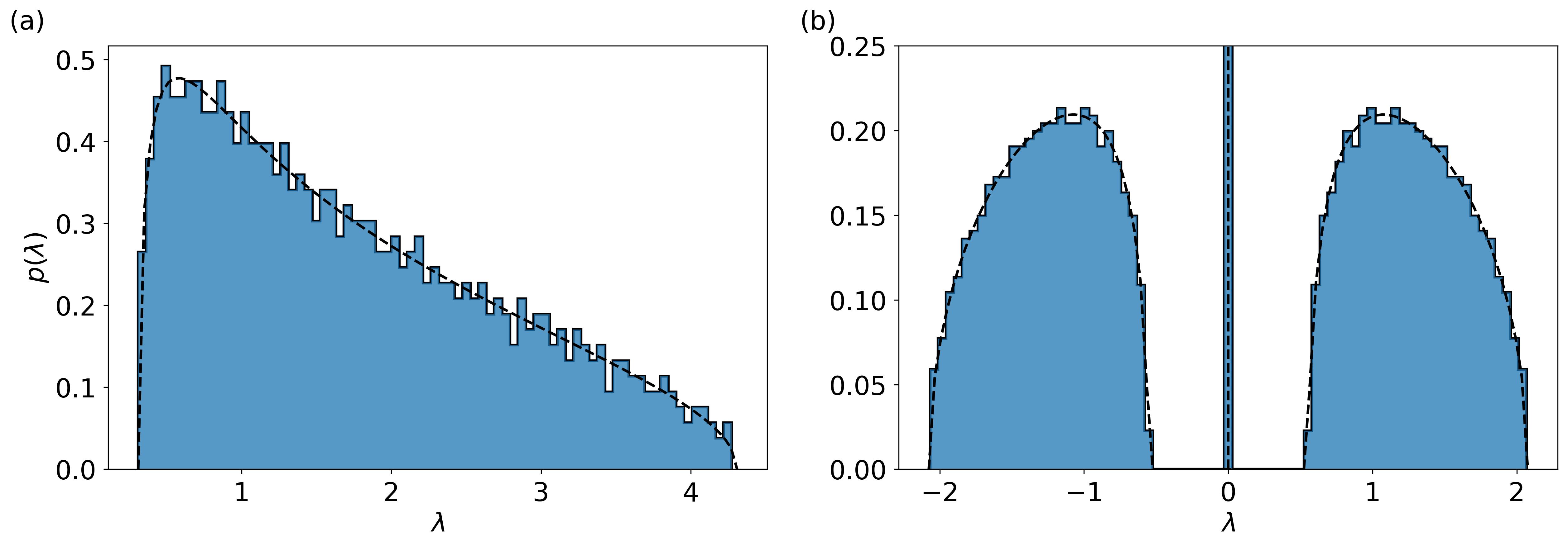}
  \captionof{figure}{
  (a) The eigenvalue distribution for a random draw of $W W^T$ for $N_v = 1000$, $N_h = 3000$ and $\beta = \mathcal{J} = 1$. The dashed line corresponds to the large-$N$ analytic prediction given by the Marchenko-Pastur distribution.
  (b) The eigenvalue distribution for the coupling matrix $J$ constructed using the $W$ matrix used in (a). The dashed line corresponds to the large-$N$ analytic prediction, which may be obtained from the Marchenko-Pastur analytic prediction and Eq.~\ref{eq:JtoWWT}.}
  \label{fig:Marchenko_Pastur}
\end{figure}

\section{Spin-glasses on Erdős-R\'{e}nyi random graphs} \label{sec:randomgraph}
As a final example, we examine another prototypical spin-glass model by placing spins on Erdős-R\'{e}nyi random graphs. We will consider $J_{ij}$ to be a Bernoulli random variable, by which we mean that
\begin{equation}
J_{i<j} = \begin{cases}
      \mathcal{J} & \text{with probability $p$} \\
      0 & \text{with probability $1-p$}
   \end{cases}
\end{equation}
As before, the $i>j$ entries are fixed to be $J_{ji} = J_{ij}$, and the diagonal entries are zero. Thus, $J$ is proportional to the adjacency matrix for an Erdős-R\'{e}nyi random graph \cite{erdHos1960evolution}.

In the limit where $1 \gg p \gg N^{-1/3}$, the eigenvalues of $J$ have been shown to obey a semi-circle law \cite{erdHos2013spectral}. The first $(N-1)$ eigenvalues form the bulk of the spectrum, and lie within the strip $[-R,R]$, with
\begin{equation}
    R = \frac{2 q \mathcal{J}}{\gamma} \,,
\end{equation}
where $q^2 := p N$ and $\gamma := (1- q^2/N)^{-1/2}$.\footnote{Although $\gamma \rightarrow 1$ as $N \rightarrow \infty$, we include it here to bring the analytic prediction closer to the numerical results we find for finite $N$.} The largest eigenvalue is separated from the bulk eigenvalues due to the fact that the matrix has a non-zero mean, and the average value is
\begin{equation}
    \mathbb{E}_{J} \lambda_{N}(J) = \mathcal{J} (\gamma^{-2} + q^2)  \,.
\end{equation}
Thus, in this limit the mean-field and convex temperatures may be worked out to be:
\begin{equation}
    T_{\text{mean-field}} = \left(\gamma^{-2} + q^2 \right) \mathcal{J} \,, \qquad  T_{\text{convex}} = \left(\gamma^{-1} + q \right)^2 \mathcal{J} \,.
\end{equation}
Note that in this example $T_{\text{convex}} \neq 2\, T_{\text{mean-field}}$, which is due to the fact that the eigenvalues of $J$ do not lie in a symmetric interval. The first $(N-1)$ eigenvalues do, but the largest eigenvalue spoils the symmetry.

\begin{figure}
\centering
  \includegraphics[width=\linewidth]{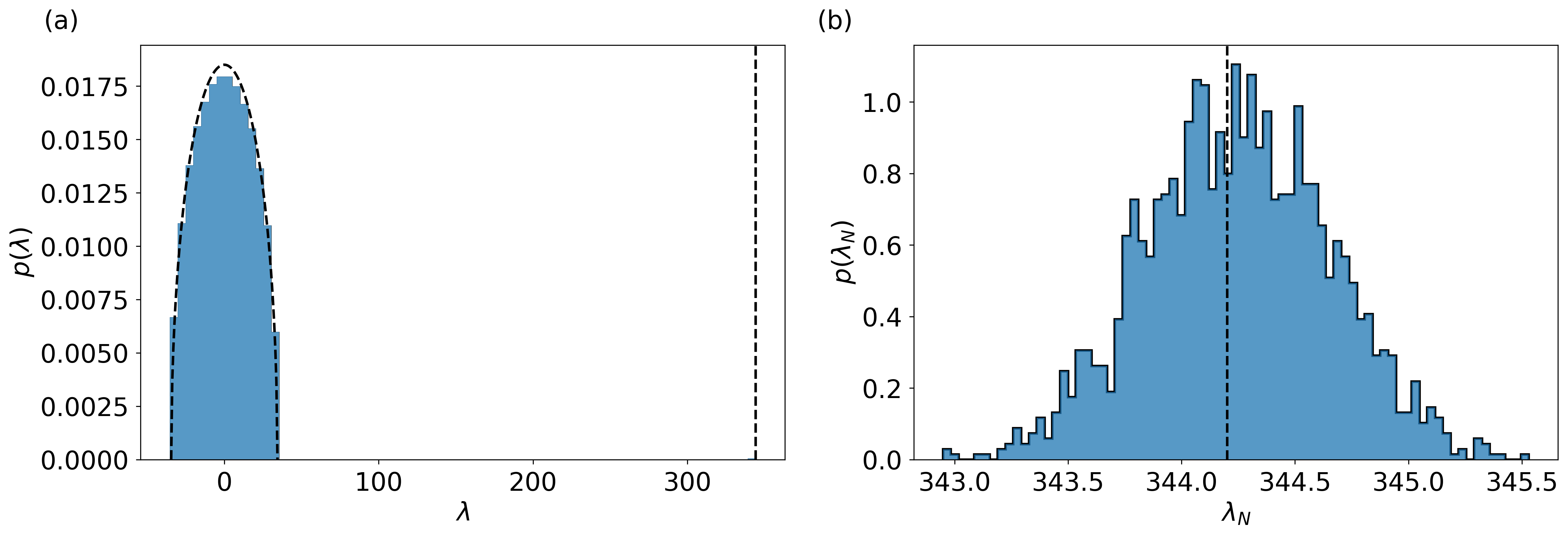}
  \captionof{figure}{
  (a) The eigenvalue distribution for a random draw the coupling matrix from an Erdős-R\'{e}nyi distribution with $N = 4000$, $q = 1.2 N^{1/3}$, and $\beta \mathcal{J} = 1$. The analytic prediction of \cite{erdHos2013spectral} is depicted by a dashed line. There is a single eigenvalue separated from the bulk at the location of the vertical dashed line.
  (b) The largest eigenvalue of the coupling matrix $J$ for 2000 draws of the Erdős-R\'{e}nyi random graph, with $N = 4000$, $q = 1.2 N^{1/3}$, and $\beta \mathcal{J} = 1$. The dashed line represents the analytic prediction.
  }
  \label{fig:ER_eigs}
\end{figure}

\section{Discussion \label{sec:conclusion}}
In this work, we have presented a probability density theory of general spin-glass systems. This formulation builds off of the continuous relaxation method of \cite{zhang2012continuous}, who introduced the method as a way to apply Hamiltonian Monte Carlo to energy-based models defined over discrete variables. Our original motivation for this work was to adopt a similar approach, and use their continuous relaxation method as a way to recast spin-glasses in terms of continuous variables so that normalizing flows could then be trained to model and better sample spin-glass physics. This is done in our companion paper \cite{hartnett2020selfsupervised}, and we have instead devoted this work to the task of further developing the continuous relaxation method for spin-glass physics and characterizing their physical properties in this picture.

One of the main advantages of our probability density formulation is the fact that $\mathcal{H}_{\beta}(x)$ furnishes a geometric encoding of complex spin-glass distributions. The critical points are of paramount importance for understanding this encoding. For example, the topology of the manifold defined by considering all points with energy equal to or below some reference value, i.e. ${\mathcal{M}^E := \{ x: \mathcal{H}_{\beta}(x) \le E\}}$, can be determined using knowledge of the critical points and Morse theory \cite{milnor2016morse}. As the energy is varied, the topology of $\mathcal{M}^E$ is unchanged unless a critical point is crossed, in which case the topology changes by the addition of a $\gamma$-cell, where $\gamma$ is the index of the critical point. Unfortunately, in none of the examples considered here were we able to identify all the critical points. It may be possible to make progress on this for the mean-field models like the SK model.

While we have not fully mapped out the geometry of the energy landscape, our results combine to form an interesting picture. For $T > T_{\text{convex}}$, the Hamiltonian density is convex and there is a single global minimum at the origin. As the temperature is lowered below $T_{\text{convex}}$, the energy landscape becomes non-convex. We have shown that this transition is accompanied by 1) the appearance of a pair of new critical points which split off from the origin, and 2) the the transition of the origin from being a minimum to being an unstable critical point (i.e. the Hessian at $x=0$ acquires a negative eigenvalue). In all the examples we considered $T_{\text{convex}}$ exceeds the critical temperature of any phase transition which may be present. We suspect that this statement is universally true unless $\Delta = 0$ and the phase transition is able to be captured by naive mean-field theory, since in that case $T_{\text{mean-field}} = T_{\text{convex}}$, otherwise $T_{\text{mean-field}} < T_{\text{convex}}$. Assuming that the model possesses a spin-glass phase transition, then as the temperature is lowered further, as far as $T \le T_{\text{crit}} < T_{\text{convex}}$, a number of metastable states appear. For the SK model, these are given by solutions of the TAP equation, and they are clearly distinct from the critical points. However, in the limit $T \rightarrow 0$, the equation satisfied by the metastable states simplifies to $x = \text{sgn}(J x)$, and we showed that solutions of this equation are also critical points of the zero-temperature Hamiltonian density. Therefore, at zero-temperature the energy landscape of the Hamiltonian density encodes the metastable states of the spin-glasss.

Our probability density formulation is quite general and applies to any system of Ising spin variables whose Hamiltonian consists of just 1- and 2-spin interactions (i.e. the family represented by Eq.~\ref{eq:Hspinglass}). It would be interesting to extend our approach to more general Hamiltonians consisting of higher spin couplings. However, it is worth noting that even the more restrictive family of Hamiltonians considered here is already universal. In particular, it includes the energy function of RBMs, which were shown to be universal approximators capable of approximating any discrete distribution to arbitrary accuracy, provided there are enough hidden units \cite{le2008representational}. Therefore, our formulation can be used to convert arbitrarily complex distributions over discrete variables into continuous probability densities. Moreover, a whole suite of algorithms and numerical methods designed for systems of continuous variables may now be applied to discrete problems. It will be interesting to understand what properties of a spin-glass system determine whether a continuous or discrete representation needs to be employed. We hope to report progress on this question in the future.

\subsubsection*{Acknowledgments}
We would like to thank S. Isakov, D. Ish, K. Najafi, and E. Parker for useful discussions and comments on this manuscript. We would also like to thank the organizers of the workshop  \textit{Theoretical Physics for Machine Learning}, which took place at the Aspen Center for Physics in January 2019, for stimulating this collaboration and project.

\appendix

\section{Convexity of the Hamiltonian density \label{sec:convex}}
In this appendix we prove that the Hamiltonian density is convex if and only if $T > T_{\text{convex}}$, with $T_{\text{convex}} = \lambda_N(\tilde{J})$. Using our conventions, in \cite{zhang2012continuous} it was proven that $p(x)$ for $\beta = 1$ is log-concave if and only if the eigenvalue spectrum of $\tilde{J}$ is sufficiently narrow, by which we mean
\begin{equation}
    \label{eq:narrowband}
    0 < \lambda_i( \tilde{J} ) < 1 \,, \quad \forall i \in \{1, ..., N\} \,.
\end{equation}
Note that the left-hand inequality of Eq.~\ref{eq:narrowband} is true by construction, since the shift $\Delta$ was chosen so as to make $\tilde{J} $ positive definite. Thus, the spectrum will be narrow if we additionally have that $\lambda_N(\tilde{J}) < 1$. Here we repeat the proof of \cite{zhang2012continuous} for the case of Ising spin variables $s \in \{-1,1\}^N$ and our parametrization. Also, rather than considering the log-concavity of $p(x)$, we shall instead consider the equivalent condition of the convexity of $\mathcal{H}_{\beta}(x)$.

In the proof we will make use of the relations $\lambda_1(-M) = -\lambda_N(M)$ and $\lambda_1(M^{-1}) = \lambda_N(M)^{-1}$, which hold for $M = \tilde{J}$ and for $M = S(x)$, and we will also use the following eigenvalue inequalities for two matrices $A$, $B$:
\begin{equation}
    \lambda_1(A) + \lambda_1(B) \le \lambda_1(A + B) \le \lambda_1(A) + \lambda_N(B) \,.
\end{equation}
We will also need the Hessian $K_{ij}(x) := \partial_i \partial_j \mathcal{H}_{\beta}(x)$, which is
\begin{equation}
    K(x) = \tilde{J} -\beta \tilde{J} S(x) \tilde{J} \,,
\end{equation}
where $S(x)$ is the diagonal matrix given by $S_{ij}(x) := \text{sech}^2( \beta \tilde{h}_i(x) ) \delta_{ij}$. We will find it useful to work with the matrix $\tilde{K}(x) := \beta^{-1} \tilde{J}^{-1} K(x) \tilde{J}^{-1}$, which is equal to
\begin{equation}
    \tilde{K}(x) = (\beta \tilde{J})^{-1} - S(x) \,.
\end{equation}
Since $K$ and $\tilde{K}$ are congruent, they have the same numbers of positive, negative, and zero eigenvalues according to Sylvester's Law of Inertia \cite{sylvester1852xix}.
\newline
\begin{proposition}
\label{prop:narrowband}
The Hamiltonian density $\mathcal{H}_{\beta}(x)$ is convex if and only if $\beta \tilde{J}$ has a narrow spectrum, by which we mean $\lambda_{N}(\beta \tilde{J}) < 1$. This is equivalent to $T_{\text{convex}} < T$.
\end{proposition}

\begin{proof}
For the forward direction, assume that $\lambda_{N}(\beta \tilde{J}) < 1$. Then
\begin{equation}
    \lambda_{1} \left( \tilde{K}(x) \right) = \lambda_{1} \left( (\beta \tilde{J})^{-1} - S(x) \right) \ge \lambda_{1} \left( (\beta \tilde{J})^{-1} \right) + \lambda_{1}\left(-S(x) \right) \ge \lambda_{N} ( \beta \tilde{J} )^{-1} - 1 > 0 \,.
\end{equation}
where we have used $\lambda_1(-S(x)) \ge -1$ in the penultimate step, and the last inequality follows by assumption. Since the smallest eigenvalue of the Hessian $\tilde{K}(x)$ is everywhere positive, the Hamiltonian density is convex.

To prove the reverse direction, assume that $0 < \inf_x \lambda_{1}(\tilde{K}(x))$ and let $x^* = - \tilde{J}^{-1} h$. Then,
\begin{align}
    0 &< \inf_x \lambda_{1}\left( \tilde{K}(x) \right) \le \lambda_{1}\left( \tilde{K}(x^*) \right) \le \lambda_{1}\left( (\beta \tilde{J})^{-1} \right) + \lambda_{N}\left( -S(x^*) \right) =  \lambda_{N}(\beta \tilde{J} )^{-1} - 1 \,,
\end{align}
Therefore, $\lambda_{N}(\beta \tilde{J}) < 1$.
\end{proof}

\section{High-temperature expansion of the SK model \label{sec:highTSK}}
According to Eq.~\ref{eq:Zx}, the discrete partition function is related to the continuous partition function via:
\begin{equation}
    Z_s = \frac{e^{-\frac{N \beta \Delta}{2}} Z_x}{\sqrt{(2\pi)^N \det( (\beta \tilde{J})^{-1})}} \,.
\end{equation}
The denominator is just the normalization of a Gaussian distribution with zero mean and covariance matrix $(\beta \tilde{J})^{-1}$. Similarly, $Z_x$ may also be written in terms of an un-normalized Gaussian integral with the same mean and covariance:
\begin{equation}
    Z_x = \int \mathrm{d}^N x \, e^{-\beta \mathcal{H}_{\beta}(x)} = \int \mathrm{d}^N x \, e^{-\frac{\beta}{2} x^T \tilde{J} x} \prod_i 2\cosh\left(\beta (\tilde{J} x)_i \right) \,.
\end{equation}
Thus, after rescaling $x \rightarrow \beta^{-1/2} x$ the partition function may be written as
\begin{equation}
    Z_s = 2^N e^{-\frac{N \beta \Delta}{2}} \left\langle \prod_i \cosh\left( \beta^{1/2} (\tilde{J} x)_i \right) \right\rangle_0 \,.
\end{equation}
where $\langle \cdot \rangle$ denotes an average with respect to the Gaussian with covariance matrix $\Sigma = \tilde{J}^{-1}$. The expansion around $\beta = 0$ may now be carried out. Explicitly, to fourth order we have
\begin{equation}
	Z_s = 2^N e^{-\frac{N \beta \Delta}{2}} \left\langle \prod_i \left(1 + \frac{\beta}{2} (\tilde{J} x)_i^2 + \frac{\beta^2}{24} (\tilde{J} x)_i^4 + \frac{\beta^3}{720} (\tilde{J} x)_i^6 + \frac{\beta^4}{40320} (\tilde{J} x)_i^8 + \mathcal{O}(\beta^{5}) \right) \right\rangle_0 \,,
\end{equation}
and expanding out the product yields
\begin{align}
    \label{eq:big_expansion}
	Z_s  2^{-N} e^{\frac{N \beta \Delta}{2}} = 1 &+ \frac{\beta}{2} \sum_i \langle (\tilde{J} x)_i^2 \rangle_0 + \beta^2 \Bigg( \frac{1}{24} \sum_i \langle (\tilde{J} x)_i^4 \rangle_0 + \frac{1}{8} \sum_{[ij]} \langle (\tilde{J} x)_i^2 (\tilde{J} x)_j^2 \rangle_0 \Bigg) + \\
	&+ \beta^3 \Bigg( \frac{1}{720} \sum_i \langle (\tilde{J} x)_i^6 \rangle_0 + \frac{1}{48} \sum_{[ij]} \langle (\tilde{J} x)_i^2 (\tilde{J} x)_j^4 \rangle_0 + \frac{1}{48} \sum_{[ijk]} \langle (\tilde{J} x)_i^2 (\tilde{J} x)_j^2 (\tilde{J} x)_k^2 \rangle_0 \Bigg) \nonumber \\
	&+ \beta^4 \Bigg( \frac{1}{40320} \sum_i \langle (\tilde{J} x)_i^8 \rangle_0 + \frac{1}{1440} \sum_{[ij]} \langle (\tilde{J} x)_i^2 (\tilde{J} x)_j^6 \rangle_0 + \frac{1}{1152} \sum_{[ij]} \langle (\tilde{J} x)_i^4 (\tilde{J} x)_j^4 \rangle_0 \nonumber \\
	& \qquad \quad + \frac{1}{192} \sum_{[ijk]} \langle (\tilde{J} x)_i^2 (\tilde{J} x)_j^2 (\tilde{J} x)_k^4 \rangle_0 + \frac{1}{384} \sum_{[ijkl]} \langle (\tilde{J} x)_i^2 (\tilde{J} x)_j^2 (\tilde{J} x)_k^2  (\tilde{J} x)_l^2 \rangle_0  \Bigg) + \mathcal{O}(\beta^{5}) \,. \nonumber
\end{align}
Here, the notation $[i_1 i_2 ... i_n]$ means that all the indices are distinct.

The Gaussian integrals can be done via Wick contractions, each of which brings in a factor of $\tilde{J}^{-1}$. Working out the first few terms explicitly, one finds a proliferation of terms involving powers of $\Delta$. These terms may be summed to give $e^{\frac{N \beta \Delta}{2}}$, which cancels the term on the RHS of Eq.~\ref{eq:big_expansion}. Of course, this had to be the case because $Z_s$ is independent of $\Delta$. Carrying out the expansion to fourth order, we find that
\begin{equation}
    \ln Z_s = N \ln 2 + \frac{\beta^2}{4} \sum_{[ij]} J_{ij}^2 + \frac{\beta^3}{6} \sum_{[ijk]} J_{ij} J_{jk} J_{ki} + \frac{\beta^4}{8} \left( \sum_{[ijkl]} J_{ij} J_{jk} J_{kl} J_{li} - \frac{1}{3} \sum_{[ij]} J_{ij}^4 \right) + \mathcal{O}(\beta^5) \,.
\end{equation}
Taking the disorder average yields:
\begin{equation}
    \label{eq:expansionpartial}
    \langle \ln Z_s \rangle_J = N \left(\ln 2 +  \frac{1}{4} (\beta \mathcal{J})^2\right) - \frac{1}{4}\left( (\beta \mathcal{J})^2 + \frac{1}{2} (\beta \mathcal{J})^4 + \mathcal{O}(\beta^5) \right) + \mathcal{O}(N^{-1})\,.
\end{equation}

With enough effort, the expansion may be extended to arbitrary order. To facilitate comparison with the result obtained by expanding in terms of the discrete variables performed by TAP \cite{thouless1977solution}, we will note that TAP did not keep track of each term in the expansion. They worked out all terms which contribute at $\mathcal{O}(N)$, and at sub-leading order $\mathcal{O}(1)$ they restricted their attention to only those terms in the series which contribute to a singularity at the spin-glass phase transition. Concretely, they found:
\begin{equation}
    \label{eq:TAPresult}
    \langle \ln Z_s \rangle_J = N\left(\ln 2 + \frac{1}{4} (\beta \mathcal{J})^2 \right) + \frac{1}{4} \ln\left(1-\beta^2 \mathcal{J}^2 \right) + \text{(non-singular)} + \mathcal{O}(N^{-1}) \,.
\end{equation}
Our calculation has already reproduced the extensive term. Next we will show that the expansion in terms of the $x$ variables also matches the singular logarithm term. For this purpose we will restrict attention to just those terms which involving $n$ distinct copies of $(\tilde{J} x)^2$. Using the notation $A \supset B$ to indicate that the expansion for $A$ contains the expansion $B$, the series restricted to these terms is:
\begin{equation}
    Z_s \supset 2^N \sum_{n=3}^{\infty} \frac{\beta^{n}}{2^n n! } \sum_{[i_1 ... i_n]} \langle (J x)_{i_1}^2 (J x)_{i_2}^2 ... (J x)_{i_n}^2 \rangle_0
\end{equation}
There are $(2n-1)!!$ different Wick contractions to consider. Of these, there are $(2n-2)!!$ contractions that avoid pairing $x$'s connected through a coupling matrix. Thus the contribution of the connected cyclic terms at each order is given by:
\begin{equation}
    Z_s \supset 2^N \sum_{n=3}^{\infty} \frac{\beta^{n}}{2n} \sum_{[i_1 ... i_n]} J_{i_1 i_2} J_{i_2 i_3} ... J_{i_n i_1} \,.
\end{equation}
The series contains additional terms that result from other contractions, but these are either sub-leading in $1/N$ or do not contribute to the singularity. The contribution of the cyclic terms above may be represented diagrammatically as regular $n$-sided polygons diagrams, where each side represents a factor of the coupling matrix
\begin{equation}
    Z_s \supset \Ngon{3}+\Ngon{4}+\Ngon{5}+\Ngon{6} + \cdots
\end{equation}

The next step is to take the logarithm and perform the disorder average. The logarithm introduces additional terms at each order, although many of these vanish at leading order in $1/N$ after taking the disorder average. Among the terms which survive are the squares of the above polygon terms:
\begin{equation}
	\langle \ln Z_s \rangle_J \supset - \frac{1}{2} \sum_{n=3}^{\infty} \frac{\beta^{2n}}{(2n)^2} \left\langle \left( \sum_{[i_1 ... i_n]} J_{i_1 i_2} J_{i_2 i_3}... J_{i_n i_1} \right)^2 \right\rangle_J \,.
\end{equation}
The disorder average may also be performed using Wick contractions. The only non-vanishing contractions are those where each distinct factor of the coupling $J_{ij}$ appears squared. Diagrammatically, this means that the contribution of these terms corresponds to the double-sided regular polygons:
\begin{equation}
    \langle \ln Z_s \rangle_J \supset \Ngondouble{3} + \Ngondouble{4} + \Ngondouble{5} + \Ngondouble{6} + \cdots
\end{equation}
To count the number of each term, note that there are two cyclic groupings of matrices which must be contracted with one another: $J_{i_1 i_2} J_{i_2 i_3} ... J_{i_n i_1}$ and $J_{j_1 j_2} J_{j_2 j_3} ... J_{j_n j_1}$. Each $J$ from the first group will be contracted with a $J$ from the second. With no loss of generality, the ordering of the first cycle can be fixed. There are then $(n-1)!$ ways to order the second cycle. However, this over counts by a factor of 2 because the direction of the cycle is irrelevant. So, the symmetry factor for the $n$-th diagram is $(n-1)!/2$. There is also a factor of $\binom{N}{n}$ corresponding to the number of choosing $n$ distinct sites to form a cycle. The end result is that the double-sided polygons give a contribution of:
\begin{equation}
    \langle \ln Z_s \rangle_J \supset -\frac{1}{4} \sum_{n=3}^{\infty} (n-1)! \binom{N}{n} (\langle J_{ij}^2 \rangle_J)^n = - \frac{1}{4} \sum_{n=3}^{\infty} \frac{(\beta \mathcal{J})^{2n}}{n} + \mathcal{O}(N^{-1}) \,.
\end{equation}
This series just corresponds to $-\ln(1-\beta^2 \mathcal{J}^2)/4$, minus the $n=1$ and $n=2$ terms. Adding this result to the previous result of Eq.~\ref{eq:expansionpartial} reproduces the expression TAP found in \cite{thouless1977solution}, which we have reproduced in Eq.~\ref{eq:TAPresult}.

Therefore, we have found that, regardless of which formulation is used, the disorder-averaged high-temperature expansion produces an extensive $\mathcal{O}(N)$ result which is valid above the spin-glass phase transition temperature. At sub-leading order $\mathcal{O}(1)$ the expansion also contains an infinite number of cyclic terms which diagrammatically correspond to regular polygons. These terms may be re-summed to find a contribution which becomes singular at the phase transition, indicating that the perturbative expansion has broken down. We see no indication that $T_{\text{convex}}$ has any particular significance whatsoever in the partition function. Indeed, $T_{\text{convex}}$ depends on $\Delta$, a parameter introduced as part of the definition of the continuous formulation, whereas $Z_s$ does not. Lastly, we note that the two partition functions $Z_s$ and $Z_x$ are proportional, and that the constants of proportionality are completely well-behaved at $T = T_{\text{convex}}$. Thus, we can conclude that the convex/non-convex transition does not correspond to any sort of phase transition or non-analyticity in either partition function.

\section{2d Ising model}\label{sec:examples}
As an additional example, we consider the phase transition of the well-studied ferromagnetic Ising model defined over the 2-dimensional square lattice with periodic boundary conditions. The probability density formulation of this model was used in \cite{li2018neural} as the first step towards modeling the system with normalizing flows \cite{rezende2015variational, dinh2016density} near the paramagnetic/ferromagetic phase transition.

In this case the eigenvalues may be worked out analytically. For a $d$-dimensional hypercubic lattice with $L$ spins per dimension, the eigenvalues are given by\footnote{See for example Sec. 2.2 of \cite{altland2010condensed}.}
\begin{equation}
    \lambda(J) = 2 \mathcal{J} \sum_{\mu =1}^d \cos\left( \frac{2\pi}{L} n_{\mu} \right) \,, \quad n_{\mu} \in \{0,1,...,L-1\} \,.
\end{equation}
Where $\mathcal{J}$ is the bond strength. Thus,
\begin{equation}
    \lambda_N(J) = 2 d \mathcal{J} \,, \qquad
    \lambda_1(J) =
    \begin{cases}
      -2 d \mathcal{J} & \text{$L$ even} \\
      2 d \mathcal{J} \cos\left(\pi \left(1 - L^{-1} \right) \right) & \text{$L$ odd}
   \end{cases}
\end{equation}
In the large-$L$ limit, the difference between even and odd $L$ vanishes, and the eigenvalues lie within the symmetric interval $[-R,R]$ with $R = 2 d \mathcal{J}$. As a result, for $d=2$,
\begin{equation}
    T_{\text{mean-field}} = 4 \,\mathcal{J} \,, \qquad T_{\text{convex}} = 2 \,T_{\text{mean-field}} \,.
\end{equation}
Both of these are greater than the critical temperature, which is well known to be
\begin{equation}
    T_{\text{crit}} = \frac{2 \mathcal{J}}{\ln\left(1 + \sqrt{2}\right)} \approx 2.2692 \mathcal{J} \,.
\end{equation}
Thus, as the temperature is lowered the Hamiltonian density becomes non-convex well before the phase transition.

\bibliographystyle{JHEP}
\bibliography{refs}

\end{document}